\documentclass{article}
\usepackage{amsmath}
\usepackage{amscd}
\usepackage{graphicx}
\usepackage{latexsym}
\usepackage{amsmath,amsthm,amsfonts,amssymb,cite%
}
\usepackage{color}
\usepackage{extarrows}
\usepackage{latexsym}
\usepackage{amsmath}

\theoremstyle{remark}

\theoremstyle{definition}
\newtheorem{definition}{Definition}[section]

\newtheorem{prop}{Proposition}[section]

\newtheorem*{rem}{Remark}
\renewcommand\dagger*
\renewcommand\Theta{\boldsymbol\theta}
\numberwithin{equation}{section}
\begin{document}
\hoffset = -2.4truecm \voffset = -2truecm
\renewcommand{\baselinestretch}{1.2}
\newcommand{\mb}{\makebox[10cm]{}\\ }
\title{The Cauchy two-matrix model, C-Toda lattice and CKP hierarchy}
\author{Chunxia Li\thanks{School of Mathematical Sciences,
Capital Normal University,
Beijing 100048, CHINA}, Shi-Hao Li{\thanks{LSEC, Institute of Computational Mathematics and Scientific Engineering Computing, AMSS, Chinese Academy of Sciences, P.O.Box 2719, Beijing 100190, CHINA
} \thanks{Department of Mathematical Sciences, University of the Chinese Academy of Sciences, Beijing, CHINA} \thanks{Corresponding author: Shi-Hao Li (lishihao@lsec.cc.ac.cn)}
}}
\date{}
\maketitle

\begin{abstract}
This paper mainly talks about the Cauchy two-matrix model and its corresponding integrable hierarchy with the help of orthogonal polynomials theory and Toda-type equations. Starting from the symmetric reduction of Cauchy biorthogonal polynomials, we derive the Toda equation of CKP type (or the C-Toda lattice) as well as its Lax pair by introducing time flows. Then, matrix integral solutions to the C-Toda lattice are extended to give solutions to the CKP hierarchy which reveals the time-dependent partition function of the Cauchy two-matrix model is nothing but the $\tau$-function of the CKP hiearchy. At last, the connection between the Cauchy two-matrix model and Bures ensemble is established from the point of view of integrable systems.
\end{abstract}

\vskip 0.5cm
\tableofcontents

%
%

\section{Introduction}
The theory of matrix models have been an incredibly fertile ground for the intriguing connections between theoretical physics, statistics, analysis and combinatorics in the past thirty years. It has been continuously studied since matrix models provide possible models for non-perturbative string theory and two-dimensional gravity \cite{BK,GM}. One of the most attractive features is their connections with integrable theory which indicates partition functions of certain matrix models are $\tau$-functions of integrable systems. {Moreover, the partition function of the matrix model can be cancelled by the Borel Virasoro algebra, which provides us a way to reformulate the matrix models in terms of more invariant terms and obtain the corresponding Virasoro constraints. This approach can also be obtained in terms of vertex operator in \cite{Adler95,DKM}.}

In 2009, Bertola et al. proposed a new matrix model called the Cauchy two-matrix model \cite{BGS3}. The Cauchy two-matrix model is an analogue to the Itzykson-Zuber-Harish-Chandra (IZHC) model. Different from the IZHC model whose metric between two matrices is
\begin{align*}
d\mu(M_1,M_2)=dM_1dM_2e^{-N\mbox{Tr}(V_1(M_1)+V_2(M_2)-M_1M_2)},
\end{align*}
the Cacuhy two-matrix model is equipped with the metric
\begin{align*}
d\mu(M_1,M_2)=dM_1dM_2\frac{\alpha(M_1)\beta(M_2)}{\det(M_1+M_2)^N}
\end{align*}
with $M_1$, $M_2$ being positive Hermitian matrices of size $N$.
As is known, the partition function of IZHC model is related to the KP hierarchy (or $2$d-Toda hierchy) \cite{DKM}, so it is natural for us to ask what the corresponding integrable hierarchy is related to the Cauchy two-matrix model.

For this purpose, we will use the method based on the orthogonal polynomials theory \cite{DKM,GMMMO}. This method is based on the fact that the Toda-type equation, whose $\tau$-function is the most general in an integrable hierarchy, can be derived from orthogonal polynomial theory. According to Sato theory, if we choose suitable higher-order time flows, then the $\tau$-function of the Toda-type equation can be generalized to the corresponding integrable hierarchy. Some famous examples include the $2$d-Toda lattice and the KP hierarchy \cite{DKM}, the $1$d-Toda lattice and the KdV hierarchy \cite{GMMMO}, the Toeplitz lattice and the mKdV hierarchy \cite{AM03} and so on. Therefore, the connections between matrix models and integrable hierarchies can be transformed into the problems of relationships between orthogonal polynomials and Toda-type equations.

As the average characteristic polynomials of the Cauchy two-matrix model, the Cauchy biorthogonal polynomials (CBOPs for brevity) are taken into consideration \cite{BGS3}. This family of polynomials are firstly proposed to solve the Hermite-Pad\'e type approximation problems associated with the inverse spectral problem for the peakon solution of Degasperis-Procesi (DP) equation \cite{BGS1,BGS2}. Later on, the Christoffel-Darboux identity is established and the determinant point process related to the Cauchy two-matrix model is found. In contrast with the standard orthogonal polynomials, a charming character of CBOPs is the four-term recurrence relation which can characterize the corresponding Riemann-Hilbert problem and provide a spectral problem for integrable hierarchies. In \cite{HMST}, Miki and Tsujimoto derived two discrete integrable systems from the discrete spectral transformations of CBOPs.

In this paper, we will give an evolutionary perspective and derive the corresponding semi-discrete lattice by imposing a time deformation in measure or weight function. The equation, which is produced by CBOPs, is called the Toda lattice of CKP type (or C-Toda lattice) in this paper because it enjoys the same $\tau$-function as that of CKP hierarchy. Furthermore, with a suitable choice of higher-order time flow, we find that the partition function of the Cauchy two-matrix model indeed can act as the $\tau$-function of the CKP hierarchy and give the answer to the question which is mentioned above.

Besides, Bertola et al. \cite{BGS3} realized that there may exist some connections between the Cauchy two-matrix model and Bures ensemble on the level of correlation functions. A direct proof of the relation between Cauchy and Bures ensemble was given in \cite{FK} based on the relations between determinants and Pfaffians. Since as was shown in \cite{HL}, the time-dependent partition function of Bures ensemble can be viewed as the $\tau$-function of the BKP hierarchy and in this paper we have shown that the partition function of the Cauchy ensemble can be viewed as the $\tau$-function of the CKP hierarchy, then another natural question is whether there is a correspondence between the BKP hierarchy and the CKP hierarchy that can associate Bures ensemble with Cauchy ensemble? An affirmative answer is given in the Hirota's book \cite{HRN} for the correspondence between Sawada-Kotera equation and Kaup-Kuperschmidt equation which are the first members of the reduced BKP and CKP hierarchy, respectively. Therefore, we will use the reduction theory proposed by Jimbo and Miwa \cite{JM} in the matrix integral method and present the modified Kaup-Kuperschmidt equation connecting Bures ensemble with Cauchy ensemble.

This paper is organized as follows. In Section 2, we recall some known facts of CBOPs briefly and present the four-term recurrence relation for the symmetric reduction of CBOPs (sCBOPs) in details. In Section 3, we establish the C-Toda lattice by introducing continuous time variables into sCBOPs and derive its Lax pair. In Section 4, matrix integral solutions are presented to the C-Toda lattice and extended to give solutions to the CKP hierarchy. In Section 5, the connection between the Cauchy ensemble and Bures ensemble is clarified from the viewpoint of integrable systems. Section 6 is devoted to conclusions and discussions. In Appendix, direct proofs of Proposition 3.1 and Proposition 3.3 are given by using Pfaffians.

\section{Symmetric reduction of Cauchy biorthogonal polynomials}
Cauchy biorthogonal polynomials were firstly proposed by Lundmark and Szmigielski by studying the multi-peakon flows of Degasperis-Procesi equation \cite{HLJS}. Since then, CBOPs have drawn much attention not only in integrable systems but also in random matrix theory and field theory \cite{BGS}.
In this section, we will give a brief review of some known facts on CBOPs and then restrict ourselves to the symmetric reduction of CBOPs and present the corresponding four-term recurrence relation with explicit determinant expressions.

\subsection{Cauchy biorthogonal polynomials}
\theoremstyle{definition}
\begin{definition}\label{icbops}
Consider the bilinear inner product $\langle \cdot,\cdot\rangle$ defined on $\mathbb{R}[x]\times\mathbb{R}[x]\to\mathbb{R}$ by
\begin{align}
\langle f(x),g(y)\rangle=\iint_{\mathbb{R}_{+}^2}\frac{f(x)g(y)}{x+y}d\rho_1(x)d\rho_2(y),
\end{align}
where $d\rho_1(x),\, d\rho_2(y)$ are two Stieltjes measures on $\mathbb{R}_{+}$. Then the pair of the sequences of the monic polynomials $(\{p_m(x)\}_{m=0}^{\infty}$, $\{q_n(y)\}_{n=0}^{\infty})$ are called CBOPs with respect to the bilinear form $\langle\cdot,\cdot\rangle$ if they satisfy
\begin{align}
\langle p_n(x),q_m(y)\rangle=h_n\delta_{n,m}.\label{OR1}
\end{align}
\end{definition}

Denote the moments
\begin{align*}
I_{i,j}=\langle x^i,y^j\rangle=\iint_{\mathbb{R}_+^2}\frac{x^iy^j}{x+y}d\rho_1(x)d\rho_2(y),
\end{align*}
which are required to satisfy the following conditions:
\begin{enumerate}
\item The moment $I_{i,j}$ is finite for all $i,j\in \mathbb{Z}_{\ge 0}$,
\item The determinant of the moment matrix is non-zero, i.e. $\mbox{det}(\langle x^i,y^j\rangle)_{i,j\ge 0}^n\ne 0$ for all $n\in \mathbb{Z}_{\ge 0}$,
\end{enumerate}
and the constraints of the moment sequences are equal to the well-posedness of the CBOPs.

Assume that $p_n(x)$ and $q_m(y)$ are monic polynomials in the form of
\begin{align*}
&p_n(x)=x^n+a_{n,1}x^{n-1}+\cdots+a_{n,0},\\
&q_m(y)=y^m+b_{m,1}y^{m-1}+\cdots+b_{m,0}.
\end{align*}
In this setting, CBOPs can be uniquely determined based on the orthogonal relation (\ref{OR1}) \cite{BGS,X}.
Denote $\tau_{k}=\det(I_{i,j})_{i,j=0,\cdots,k-1}$. Then $p_n(x)$, $q_n(y)$ can be expressed in terms of determinants as
\begin{align*}
p_n(x)=\frac{1}{\tau_n}\begin{vmatrix}I_{0,0}&I_{0,1}&\cdots&I_{0,n}\\ \vdots&\vdots&&\vdots\\I_{n-1,0}&I_{n-1,1}&\cdots&I_{n-1,n}\\1&x&\cdots& x^n\end{vmatrix}, \quad q_n(y)=\frac{1}{\tau_n}\begin{vmatrix}I_{0,0}&\cdots&I_{0,n-1}&1\\I_{1,0}&\cdots&I_{1,n-1}&y\\ \vdots&&\vdots&\vdots\\I_{n,0}&\cdots&I_{n,n-1}&y^n\end{vmatrix}.
\end{align*}
and $h_n$ can be also written as a ratio of two determinants as
\begin{align*}
h_n=\frac{\tau_{n+1}}{\tau_n}.
\end{align*}

A remarkable feature of CBOPs is due to the rank-$1$ shift condition of the moments
\begin{align}\label{bsm}
I_{i,j+1}+I_{i,j+1}=\alpha_i\beta_j,\quad\quad \alpha_i=\int_{\mathbb{R}_+}x^id\rho_1(x),\,\beta_j=\int_{\mathbb{R}_+}y^jd\rho_2(y),
\end{align}
from which a four-term recurrence relation and related Riemann-Hilbert problem can be characterized \cite{BGS}.

\subsection{The symmetric reduction of CBOPs}
The symmetric reduction we take here is to let $d\rho_1(x)=d\rho_2(x)=\alpha(x)dx$ in Definition (\ref{icbops}) for CBOPs which means the measures are equal so that the moments are symmetric, i.e.
\begin{align*}
I_{i,j}=\iint_{\mathbb{R}_+^2}\frac{x^iy^j}{x+y}\alpha(x)\alpha(y)dxdy
\end{align*}
such that
$I_{i,j}=I_{j,i}$ and $\alpha_i=\beta_i$ in (\ref{bsm}). Hereafter, we call them the symmetric reduction of Cauchy biorthogonal polynomials and denote as sCBOPs for simplicity. For self-consistency, we shall re-prove that sCBOPs satisfy four-term recurrence relations and give explicit expressions in terms of determinants for coefficients in four-term recurrence relations.
\begin{prop}By assuming $$a_n=-\frac{\int_{\mathbb{R}_+} p_{n+1}(x)\alpha(x)dx}{\int_{\mathbb{R}_+} p_n(x)\alpha(x)dx},$$ sCBOPs satisfy the following four-term recurrence relations
\begin{align}\label{rr}
x(p_{n+1}(x)+a_np_n(x))=p_{n+2}(x)+b_np_{n+1}(x)+c_np_n(x)+d_np_{n-1}(x),\quad n\geq 0\end{align}
with $p_{-1}(x)=0$ and constants $b_n,c_n,d_n$ being uniquely determined based on the orthogonality condition.
\end{prop}

\begin{proof}
Since $\{p_n(x)\}_{n=0}^\infty$ constitute of the basis of the polynomials space $\mathbb{R}[x]$, we assume that $$x(p_{n+1}(x)+a_np_n(x))=p_{n+2}(x)+\sum_{i=0}^{n+1}\gamma_ip_i(x).$$
By noticing  $$a_n=-\frac{\int_{\mathbb{R}_+} p_{n+1}(x)\alpha(x)dx}{\int_{\mathbb{R}_+} p_n(x)\alpha(x)dx},$$
we have
\begin{align*}
\iint_{\mathbb{R}_+^2}{(p_{n+1}(x)+a_np_n(x))q_m(y)}\alpha(x)\alpha(y)dxdy=\int_{\mathbb{R}_+}(p_{n+1}(x)+a_np_n(x))\alpha(x)dx\int_{\mathbb{R}_+}q_m(y)\alpha(y)dy=0.
\end{align*}
Further we have
\begin{align*}
\langle x(p_{n+1}(x)+&a_np_n(x)),q_m(y)\rangle\\
&=\iint_{\mathbb{R}_+^2}\frac{x(p_{n+1}(x)+a_np_n(x))q_m(y)}{x+y}\alpha(x)\alpha(y)dxdy\\
&=\iint_{\mathbb{R}_+^2}(p_{n+1}(x)+a_np_n(x))q_m(y)\alpha(x)\alpha(y)dxdy-\langle p_{n+1}(x)+a_np_n(x),yq_m(y)\rangle\\
&=-\langle p_{n+1}(x)+a_np_n(x),yq_m(y)\rangle.
\end{align*}
Due to the orthogonality condition \eqref{OR1}, it is obvious that $$\left\langle p_{n+2}(x)+\sum_{i=0}^{n+1}\gamma_ip_i(x),q_m(y)\right\rangle=-\langle p_{n+1}(x)+a_np_n(x),yq_m(y)\rangle=0,\quad \text{for\, \,  $m<n-1$},$$
from which we can draw the conclusion that $\gamma_m=0$ for $m<n-1$. Therefore, we have four-term recurrence relations
\begin{eqnarray}\label{inpro}
x(p_{n+1}(x)+a_np_n(x))=p_{n+2}(x)+\gamma_{n+1}p_{n+1}(x)+\gamma_np_n(x)+\gamma_{n-1}p_{n-1}(x).
\end{eqnarray}
To avoid confusion of notations, we use \eqref{rr} instead of \eqref{inpro} as the four-term recurrence relation for sCBOPs.

Denote $w_i=\int_{\mathbb{R}_+}x^i\alpha(x)dx$ and
\begin{equation}
\sigma_n=\begin{vmatrix}I_{0,0}&\cdots&I_{0,n}\\ \vdots&\vdots&\vdots\\I_{n-1,0}&\cdots&I_{n-1,n}\\w_0&\cdots& w_n\end{vmatrix},\quad \tilde{\tau}_n=\begin{vmatrix}I_{0,0}&\cdots&I_{0,n-1}\\ \vdots&\vdots&\vdots\\I_{n-2,0}&\cdots&I_{n-2,n-1}\\I_{n,0}&\cdots& I_{n,n-1}\end{vmatrix}.\label{S2}
\end{equation}
By taking the orthogonality condition \eqref{OR1} into account, constants $a_n,b_n,c_n$ and $d_n$ in \eqref{rr} can be expressed explicitly.
From the determinant expression of $p_n(x)$, one can see that
\begin{equation*}
{\int_{\mathbb{R}_+} p_n(x)\alpha(x)dx}=\frac{\sigma_n}{\tau_n},\quad a_n=-\frac{\int_{\mathbb{R}_+} p_{n+1}(x)\alpha(x)dx}{\int_{\mathbb{R}_+} p_n(x)\alpha(x)dx}=-\frac{\sigma_{n+1}\tau_{n}}{\sigma_n\tau_{n+1}}.\label{AC}
\end{equation*}
Taking inner products of (\ref{inpro}) with $q_{n-1}(y)$, $q_n(y)$ and $q_{n+1}(y)$, respectively and comparing coefficients of both sides, we can get $b_n$, $c_n$ and $d_n$ expressed in terms of $\tau_n$, $\tilde{\tau}_n$ and $\sigma_n$ as
\begin{eqnarray}
\left\{\begin{aligned}
&b_n=-\frac{\sigma_{n+1}\tau_{n}}{\sigma_n\tau_{n+1}}+\frac{\tilde{\tau}_{n+2}}{\tau_{n+2}}-\frac{\tilde{\tau}_{n+1}}{\tau_{n+1}},\\
&c_n=-\frac{\tau_{n}\tau_{n+2}}{\tau_{n+1}^2}-\frac{\sigma_{n+1}\tau_{n}}{\sigma_n\tau_{n+1}}\left[\frac{\tilde{\tau}_n}{\tau_n}-\frac{\tilde{\tau}_{n+1}}{\tau_{n+1}}\right],\\
&d_n=\frac{\sigma_{n+1}\tau_{n-1}}{\sigma_n\tau_n}.
\end{aligned}\right.\label{CT}
\end{eqnarray}
 \end{proof}
The symmetric reduction of Stieltjes measures is firstly considered here which is helpful to find out the connection between C-Toda lattice and CKP hierarchy. Coefficients $a_n,\,b_n,\,c_n$ and $d_n$ in the four-term recurrence relation \eqref{rr} are expressed in terms of $\tau_n$, $\tilde{\tau}_n$ and $\sigma_n$. Later on, we will see that $\tilde{\tau}_n$ and $\tau_n$ are closely related after introducing time flows, which makes it possible to express $a_n,\,b_n,\,c_n$ and $d_n$ only in terms of $\tau_n$ and $\sigma_n$.

\section{The Toda equation of CKP type}
Recall that our main goal in this paper is to establish the connection between integrable hierarchy with Cauchy matrix model, which is equal to establish the relationship between CBOPs and corresponding Toda-type equation. Therefore, in this section, we will focus on how to derive the Toda equation of CKP type, or the C-Toda lattice from sCBOPs. Moreover, a Lax pair of C-Toda lattice will be obtained by introducing time deformation in measure or weight function of sCBOPs.

\subsection{The t-deformations of sCBOPs and the C-Toda lattice}\label{C-Toda}
To build the C-Toda lattice, we introduce the continuous time variable $t$ into the weight function $\alpha(x)$ such that $\alpha(x;t)$ has the form
\begin{align*}
\alpha(x;t)=\exp(V(x)+xt),
\end{align*}
where the function $V(x)$ is required to ensure the convergency of moments or equivalently, we introduce the time deformation on the measure $\rho(x;t)$ such that $$d\rho(x;t)=e^{xt}d\rho(x;0).$$ Consequently, the bimoments $I_{i,j}$ and single moments $\omega_i$ are dependent of time $t$ such that
\begin{align*}
&I_{i,j}=\iint_{\mathbb{R}_+^2}\frac{x^iy^j}{x+y}\alpha(x;t)\alpha(y;t)dxdy,\quad\omega_i=\int_{\mathbb{R}_+}x^i\alpha(x;t)dx,\\
&\frac{d}{dt}I_{i,j}=I_{i+1,j}+I_{i,j+1}=\omega_i\omega_j,\quad\quad\quad\frac{d}{dt}\omega_i=\omega_{i+1}.
\end{align*}
In this setting, $p_n(x)$ and $q_m(y)$ become time-dependent functions as well as $\tau_n$, $\tilde{\tau}_n$ and $\sigma_n$. Particularly, $\tau_n$ and $\tilde{\tau}_n$ are closely related to each other.

\begin{prop}\label{derivationtau}
The functions $\tau_n$ and $\tilde{\tau}_n$ given by \eqref{S2} satisfy the following formula
\begin{align*}
\frac{d}{dt}\tau_n=2\tilde{\tau}_n.
\end{align*}
\end{prop}
\begin{rem}
Here we would like to mention that we will give a proof of Proposition 3.1 by Pfaffian techniques in Appendix. A different proof by determinant techniques can be referred to \cite{CHL}.
\end{rem}
With the help of this formula, we can reformulate $a_n$, $b_n$, $c_n$ and $d_n$ in terms of $\tau_n$ and $\sigma_n$ as
\begin{align*}
&a_n=-\frac{\sigma_{n+1}\tau_n}{\sigma_n\tau_{n+1}},\quad \quad \ \ b_n=-\frac{\sigma_{n+1}\tau_n}{\sigma_n\tau_{n+1}}+\frac{1}{2}\frac{d}{dt}\log{\frac{\tau_{n+2}}{\tau_{n+1}}},\\
&c_n=-\frac{\tau_n\tau_{n+2}}{\tau_{n+1}^2}+\frac{1}{2}\frac{\sigma_{n+1}\tau_n}{\sigma_n\tau_{n+1}}\frac{d}{dt}\log\frac{\tau_{n+1}}{\tau_n},\quad d_n=\frac{\sigma_{n+1}\tau_{n-1}}{\sigma_n\tau_n}.
\end{align*}

As integrable systems are the compatibility condition of spectral problem and time deformation, therefore, it is necessary for us to consider the time deformations of $p_n(x;t)$, for which we have the following proposition.
\begin{prop}
For the sCBOPs, it follows
\begin{align}\label{defor}
\frac{d}{dt}p_{n+1}(x;t)+a_n\frac{d}{dt}p_n(x;t)=a_n\frac{d}{dt}(\log h_n)p_n(x;t).
\end{align}
\end{prop}
\begin{proof}
Given the following equality
\begin{align*}
\langle p_{n+1}(x;t)+a_np_n(x;t),q_m(y;t)\rangle=h_{n+1}\delta_{n+1,m}+a_nh_n\delta_{n,m},
\end{align*}
differentiating it with respect to $t$, we have
\begin{align}
\left\langle\frac{d}{dt}(p_{n+1}(x;t)+a_np_n(x;t)),q_m(y;t)\right\rangle+&\left\langle p_{n+1}(x;t)+a_np_n(x;t),\frac{d}{dt}q_m(y;t)\right\rangle\nonumber\\
&\quad=\frac{d}{dt}h_{n+1}\delta_{n+1,m}+\frac{d}{dt}(a_nh_n)\delta_{n,m},\label{der}
\end{align}
where we have used the fact $$\langle (x+y)(p_{n+1}(x;t)+a_np_n(x;t)),q_m(y;t)\rangle=0,\quad \forall m\in\mathbb{Z}_{\geq0}.$$

It is obvious that the $t$-derivative of $q_m(y;t)$ is a polynomial of degree $m-1$ in variable $y$. In the case of $0\leq m\leq n-1$, both the right hand side of \eqref{der} and the second term in the left hand side of \eqref{der} are equal to zero.
Therefore, \eqref{der} becomes
\begin{equation}\label{Der1}
\left\langle\frac{d}{dt}(p_{n+1}(x;t)+a_np_n(x;t)),q_m(y;t)\right\rangle=0,\quad 0\leq m\leq n-1.
\end{equation}
Noticing the $t$-derivative of $p_{n+1}(x;t)+a_np_n(x;t)$ is of degree $n$, we can assume
\begin{align}\label{Der2}
\frac{d}{dt}(p_{n+1}(x;t)+a_np_n(x;t))=\sum_{i=0}^n \Gamma_ip_i(x;t),
\end{align}
By substituting \eqref{Der2} into \eqref{Der1}, we conclude that $\Gamma_i=0$ for $0\leq i\leq n-1$ which implies
\begin{equation*}
\frac{d}{dt}(p_{n+1}(x;t)+a_np_n(x;t))=\Gamma_np_n(x;t).
\end{equation*}

When $m=n$, we know that
\begin{align*}
\left\langle\frac{d}{dt}(p_{n+1}(x;t)+a_np_n(x;t)),q_n(y;t)\right\rangle=\Gamma_nh_n=\frac{d}{dt}(a_nh_n).
\end{align*}
Therefore, $\Gamma_n=\frac{1}{h_n}\frac{d}{dt}(a_nh_n)$ and
\begin{align*}
\frac{d}{dt}(p_{n+1}(x;t)+a_np_n(x;t))=\frac{1}{h_n}\frac{d}{dt}(a_nh_n)p_n(x;t),
\end{align*}
or equivalently,
\begin{align}\label{Der3}
\frac{d}{dt}p_{n+1}(x;t)+a_n\frac{d}{dt}p_n(x;t)=a_n\frac{d}{dt}(\log h_n)p_n(x;t).
\end{align}
\end{proof}

By comparing the coefficients of $x^n$ in \eqref{Der3}, we find that
\begin{align}\label{t1}
-\frac{1}{2}\frac{d^2}{dt^2}\log\tau_{n+1}=a_n\frac{d}{dt}(\log h_n)
\end{align}
Moreover, if we take $m=n+1$ in (\ref{der}) and it follows
\begin{align}\label{t2}
-\frac{1}{2}a_nh_n\frac{d^2}{dt^2}\log\tau_{n+1}=\frac{d}{dt}h_{n+1}.
\end{align}
Equations (\ref{t1}) together with (\ref{t2}) can be rewritten equivalently as
\begin{align}
\left\{
\begin{aligned}
&D_{t}\tau_{n+1}\cdot\tau_{n}=\sigma_{n}^2,\\
&D_{t}^2\tau_{n+1}\cdot\tau_{n+1}=4\sigma_{n+1}\sigma_{n},
\end{aligned}\right.\label{CT}
\end{align}
where Hirota's bilinear operator $D$ is defined by $D_t^nf(t)\cdot g(t)=\frac{\partial^n}{\partial s^n}f(t+s)g(t-s)|_{s=0}$. We call \eqref{CT} the C-Toda lattice and conclude the obtained results with the following propostion.
\begin{rem}
Here we would like to remark that this method for deriving integrable systems from orthogonal polynomials can be found in many other cases and one can refer to \cite{GMMMO} for more examples, which include the derivation of Lotka-Volterra lattice (or so-called Kac-van Moerbeke lattice) and Toeplitz lattice and so on.
\end{rem}
\begin{prop}
The C-Toda lattice
\begin{align*}
D_t\tau_{n+1}\cdot\tau_n=\sigma_n^2,\quad D_t^2\tau_{n+1}\cdot\tau_{n+1}=4\sigma_{n+1}\sigma_n
\end{align*}
admits the following determinant solutions
\begin{align}\label{tau}
\tau_n=\begin{vmatrix}I_{0,0}&\cdots&I_{0,n-1}\\
\vdots&&\vdots\\
I_{n-1,0}&\cdots&I_{n-1,n-1}\end{vmatrix},\quad\quad \sigma_n=\begin{vmatrix}I_{0,0}&\cdots&I_{0,n}\\ \vdots&&\vdots\\I_{n-1,0}&\cdots&I_{n-1,n}\\w_0&\cdots& w_n\end{vmatrix}
\end{align}
with the time evolutions
\begin{align}
\frac{d}{dt}I_{i,j}=I_{i+1,j}+I_{i,j+1}=\omega_i\omega_j,\quad \frac{d}{dt}\omega_i=\omega_{i+1}.
\end{align}
\end{prop}
A direct proof of this proposition will be given in Appendix. Indeed, the function $\sigma_n$ in the C-Toda lattice (\ref{CT}) is indeed an auxiliary function which can be eliminated and equation \eqref{CT} is actually characterized by the function $\tau_n$ and is governed by the following single equation
\begin{align*}
16(D_t\tau_{n+2}\cdot\tau_{n+1})(D_t\tau_{n+1}\cdot\tau_n)=(D_t^2\tau_{n+1}\cdot\tau_{n+1})^2 .
\end{align*}

\begin{rem}
It should be mentioned that the motivation of derivation C-Toda lattice doesn't only lie in the connection with sCBOPs, but in the study of the positive flow of Degasperis-Procesi peakon equation. A detailed discussion about the relationship between C-Toda lattice and Degasperis-Procesi equation can be referred to \cite{CHL}.
\end{rem}

\subsection{Lax pair}
As an important integrable property, Lax pair of the C-Toda lattice \eqref{CT} is necessary to be given from the point of view of orthogonal polynomials. Usually, Lax pair can be always given by the four-term recurrence relationship (\ref{rr}) and the time deformations of sCBOPs (\ref{defor}). By setting $\phi_n=p_n(0;t)$, the Lax pair of the C-Toda lattice can be written as
\begin{align}\label{lax1}\left\{\begin{aligned}
&\phi_{n+2}+b_n\phi_{n+1}+c_n\phi_n+d_n\phi_{n-1}=0,\\
&\frac{d}{dt}\phi_{n+1}+a_n\frac{d}{dt}\phi_n=a_n\frac{d}{dt}{\log h_n}\phi_n.
\end{aligned}\right.
\end{align}
Unfortunately, this form of Lax pair is difficult for us to compute its compatibility condition. It pushes us to find a suitable time evolution part.
\begin{prop}
There exists a mixed spectral transformation of $p_n(x,t)$ as
\begin{align}\label{transf}
\frac{d}{dt}p_{n}(x;t)=p_{n+1}(x;t)-(x+a_{n-1}-b_{n-1})p_n(x;t)+\frac{d_n}{a_n}p_{n-1}(x;t).
\end{align}
\end{prop}

\begin{proof}
Consider the equality
\begin{align*}
\langle p_n(x;t),q_k(y;t)\rangle=0,\quad k<n
\end{align*}
and differentiate it with respect to $t$, we get
\begin{align}\label{inre}
\left\langle\frac{d}{dt}{p}_n(x;t),q_k(y;t)\right\rangle+\left\langle xp_n(x;t),q_k(y;t)\right\rangle+\left\langle p_n(x;t),yq_k(y;t)\right\rangle=0.
\end{align}

Since $\langle p_{n+1}(x;t),yq_k(y;t)\rangle=0$ for $k<n$, we have
\begin{align*}
\langle p_n(x;t),yq_k(y;t)\rangle&=\frac{1}{a_n}\langle p_{n+1}(x;t)+a_np_n(x;t),yp_k(y;t)\rangle=-\frac{1}{a_n}\langle x(p_{n+1}(x;t)+a_np_n(x;t)),q_k(y;t)\rangle,
\end{align*}
from which we have
\begin{align*}
&\left\langle \frac{d}{dt}{p}_n(x;t)+xp_n(x;t)-\frac{x}{a_n}(p_{n+1}(x;t)+a_np_n(x;t)),q_k(y;t)\right\rangle=\left\langle \frac{d}{dt}({p}_n(x;t))-\frac{x}{a_n}p_{n+1}(x;t),q_k(y;t)\right\rangle=0.
\end{align*}

By using the four-term recurrence relationship and orthogonality, we have
\begin{align*}
&\left\langle \frac{d}{dt}p_n(x;t)-\frac{x}{a_n}p_{n+1}(x;t),q_k(y;t)\right\rangle\\
&=\left \langle \frac{d}{dt}p_n(x;t)-\frac{1}{a_n}(xp_{n+1}(x;t)-p_{n+2}(x;t)),q_k(y;t)\right\rangle\\
&=\left\langle \frac{d}{dt}p_n(x;t)-\frac{1}{a_n}(b_np_{n+1}(x;t)+c_np_n(x;t)+d_np_{n-1}(x;t)-a_nxp_n(x;t)),q_k(y;t)\right\rangle\\
&=\left\langle \frac{d}{dt}p_n(x;t)+xp_n(x;t)-\frac{d_n}{a_n}p_{n-1}(x;t),q_k(y;t)\right\rangle\\
&=\left\langle \frac{d}{dt}p_n(x;t)+xp_n(x;t)-p_{n+1}(x;t)+(a_{n-1}-b_{n-1})p_n(x;t)-\frac{d_n}{a_n}p_{n-1}(x;t),q_k(y;t)\right\rangle\\
&=0.
\end{align*}

Assume $f(x)$ is a polynomial in variable $x$ and $\mbox{deg}(f(x))\leq n-1$ and satisfies $\langle f(x),q_k(y)\rangle=0$ for arbitrary $0\leq k\leq n-1$, then it is sufficient for us to have $f(x)=0$. From the last step, it is not difficult to see that the first expression in the inner product is of degree $(n-1)$ at most, thus we have
\begin{align}
\frac{d}{dt}p_n(x;t)=p_{n+1}(x;t)-(x+a_{n-1}-b_{n-1})p_n(x;t)+\frac{d_n}{a_n}p_{n-1}(x;t).\label{TP}
\end{align}
\end{proof}
\begin{rem}
Indeed, the Lax pair of the form (\ref{lax1}) can construct a matrix form of Lax form which is exhibited in \cite{CHL}. Only in the polynomials form (or so-called matrix form) can it be used to derive the compatibility of C-Toda lattice. In the wave-function form, it is so hard for one to compute its compatibility condition.
\end{rem}
\begin{rem}
In the proposition, we call the transformation (\ref{transf}) is a mixed spectral problem since it is not only a spectral problem but also evolves the time deformation part. Moreover, many types of orthogonal polynomials have this kind of mixed spectral problem. In the case of skew-orthogonal polynomials (SOPs), the recurrence relationship of SOPs is firstly given as a mixed spectral problem \cite{AHM}.
\end{rem}

By choosing $x=0$ and $\phi_n=p_n(0)$ in \eqref{rr} and \eqref{TP}, we are able to derive the Lax pair for the C-Toda lattice \eqref{CT}:
\begin{eqnarray}\left\{\begin{aligned}
&\phi_{n+2}+b_n\phi_{n+1}+c_n\phi_n+d_n\phi_{n-1}=0,\\
&\phi_{n,t}=\phi_{n+1}+(b_{n-1}-a_{n-1})\phi_n+\frac{d_n}{a_n}\phi_{n-1}.\label{LP}
\end{aligned}\right.
\end{eqnarray}
Denote the column vector $\Phi_n=(\phi_{n+1},\phi_n,\phi_{n-1})^t$. The Lax pair \eqref{LP} of the C-Toda lattice  can be rewritten in the following matrix form
\begin{align*}
\Phi_{n+1}=A_n\Phi_n,\quad \frac{d}{dt}{\Phi}_n=B_n\Phi_n,
\end{align*}
with $A_n$ and $B_n$ being $3\times3$ matrices
\begin{eqnarray*}
A_n=\begin{pmatrix}-b_n&-c_n&-d_n\\1&0&0\\0&1&0\end{pmatrix},\quad B_n=
\begin{pmatrix}
-a_n&\frac{d_{n+1}}{a_{n+1}}-c_n&-d_n\\
1&b_{n-1}-a_{n-1}&\frac{d_n}{a_n}\\
-\frac{1}{a_{n-1}}&1-\frac{b_{n-1}}{a_{n-1}}&b_{n-2}-a_{n-2}-\frac{c_{n-1}}{a_{n-1}}
\end{pmatrix}.
\end{eqnarray*}

\section{Cauchy two-matrix model and CKP hierarchy}
Toda equation, which plays an important role in modern physics, exhibits the most general algebraic and geometric structures in integrable theory.  Several important facts have been revealed in \cite{Adler95,GMMMO,MMM}. In these work, the authors demonstrated the partition function of the random matrix generated by Hermite ensemble gives rise to the KP/Toda hierarchy together with their Virasoro algebra. These work reflects that there exists a deep connection between random matrix theory/matrix model and integrable systems. In our previous work, we have also found partition functions of Gaussian orthogonal or symplectic ensembles and Bures ensemble can act as matrix integrals solutions to DKP hierarchy and BKP hierarchy, respectively. A brief review can be found in \cite{HL} and in this part we will develop along this line further and give another example of matrix models related to a hierarchy of integrable systems.

\subsection{Matrix integral solutions to the C-Toda lattice}
Before we proceed to show that the partition function of the Cauchy two-matrix model plays the role of the $\tau$-function to the CKP hierarchy, we will first find matrix integral solutions to the C-Toda lattice.
Recall that the bi-moments $I_{i,j}$ and single moments $w_i$ are
\begin{align*}
I_{i,j}=\iint_{\mathbb{R}_+^2}\frac{x^iy^j}{x+y}d\rho(x;t)d\rho(y;t),\quad \omega_i=\int_{\mathbb{R}_+}x^id\rho(x;t).
\end{align*}
According to Heine's formula \cite{Mehta81,Mehta04}, we see that $\tau_n$ and $\sigma_n$ in (\ref{tau}) can be expressed in terms of matrix integrals \cite{BGS1}, i.e.
\begin{align}
\tau_n=\det(I_{i,j})_{i,j=0,\cdots,n-1}&=\sum_{\sigma\in S(n)}\epsilon(\sigma)\int_{\mathbb{R}_+^{2n}}\prod_{j=1}^nx_j^{\sigma_j-1}y_j^{j-1}\frac{1}{x_j+y_j}d\rho(x_j;t)d\rho(y_j;t)\nonumber\\
&=\frac{1}{n!}\int_{\mathbb{R}_+^{2n}}\Delta_n(X)\Delta_n(Y)\prod_{j=1}^n\frac{d\rho(x_j;t)d\rho(y_j;t)}{x_j+y_j}\nonumber\\
&=\frac{1}{(n!)^2}
\int_{\mathbb{R}_+^{2n}}\Delta_n(X)\Delta_n(Y)\det\left(\frac{1}{x_i+y_j}\right)_{1\le i,j\le n}\prod_{i,j=1}^nd\rho(x_i;t)d\rho(y_j;t)\nonumber\\
&=\int_{\mathbb{T}}(\Delta_n(X)\Delta_n(Y))^2\prod_{i,j=1}^n\frac{1}{x_i+y_j} d\rho(x_i;t)d\rho(y_j;t),\label{heine}
\end{align}
where $\epsilon(\sigma)$ denotes the sign of the permutation $\sigma$, $\Delta_n(X)=\prod_{i<j}(x_i-x_j)$ and the integral region $\mathbb{T}\subset\mathbb{R}_+^n\times\mathbb{R}_+^n$ is given by $$\mathbb{T}=\{(x_1,\cdots,x_n;y_1,\cdots,y_n)|0<x_1<\cdots<x_n,0<y_1<\cdots<y_n\}.$$

In a similar manner, one can get
\begin{align*}
\sigma_n=
\int_{\mathbb{T}'}(\Delta_{n+1}(X)\Delta_n(Y))^2\prod_{i=1}^{n+1}\prod_{j=1}^n\frac{1}{x_i+y_j}d\rho(x_i;t)d\rho(y_j;t),
\end{align*}
and the integral region is a subset of $\mathbb{R}_+^{n+1}\times\mathbb{R}_+^n$ such that
$$\mathbb{T}'=\{(x_1,\cdots,x_{n+1};y_1,\cdots,y_n)|0<x_1<\cdots<x_{n+1},0<y_1<\cdots<y_n\}.$$

\subsection{Matrix integral solutions to the CKP hierarchy}
As is indicated in last subsection, the solution of the C-Toda lattice can be expressed as matrix integrals.
From the viewpoint of $\tau$-function theory, it is known that the $\tau$-function of Toda hierarchy is the most general $\tau$-function of each integral sytem in the corresponding hierarchy \cite{AZ}. Therefore, it is natural for us to think whether we can generalize the $\tau$-function of the C-Toda lattice to the one of the CKP hierarchy. The answer is positive. We find that the C-Toda lattice admits a $\tau$-function which is also the $\tau$-function of the CKP hierarchy and this is the reason that why we call equation (\ref{CT}) the Toda of CKP type.

Firstly, we would like to show that $\tau_n$ given before satisfies the CKP equation which is the first member of the CKP hierarchy. Consider the CKP equation
\begin{align}
&(D_{t_1}^4-4D_{t_1}D_{t_3})f\cdot f+6fg=0,\nonumber\\
&(D_{t_1}^6+144D_{t_1}D_{t_5}-80D_{t_3}^2-20D_{t_1}^3D_{t_3})f\cdot f-90 D_{t_1}^2 f\cdot g=0,\label{CKP}
\end{align}
where $g$ is actually an auxiliary function which can be eliminated from the first equation. In other words, $f$ plays the role of $\tau$-function in the CKP equation.
In what follows, we will demonstrate the CKP equation owns matrix integral solutions.
\begin{prop}\label{PPCKP}
The equation (\ref{CKP}) admits the following solution
\begin{align*}
f=\det(A)=\det\left(\int_{-\infty}^{t_1}\phi_i\phi_jd{t_1}\right)_{0\leq i,j\leq N-1},\quad g=2\left(\begin{vmatrix}0 &\Phi^T\\-\Phi_{t_1t_1t_1}& A\end{vmatrix}-\begin{vmatrix}0&\Phi_{t_1}^T\\-\Phi_{t_1t_1}&A\end{vmatrix}+\begin{vmatrix}0&0&\Phi_{t_1}^T\\0&0&\Phi^T\\-\Phi_{t_1}&-\Phi&A\end{vmatrix}\right),
\end{align*}
with $
\Phi^T=(\phi_0,\phi_1,\cdots,\phi_{N-1})$ and each $\phi_i$ $(i=0,1,\cdots,N-1)$ being an arbitrary function satisfying the dispersion relations
\begin{align*}
\frac{\partial \phi_i}{\partial t_3}=\frac{\partial^3 \phi_i}{\partial t_1^3},\quad\frac{\partial \phi_i}{\partial t_5}=\frac{\partial^5 \phi_i}{\partial t_1^5}.
\end{align*}
\end{prop}
This proposition has been proven by Wang in \cite{W}. Let us take $\phi_i$ $(i=0,1,\cdots,N-1)$ to be a single moment with the following time-dependent weight function
\begin{align}\label{sm}
\phi_i=\int_{\mathbb{R}_+}x^i\alpha(x;t)dx,\quad \alpha(x;t)=\exp(V(x)+xt)=\exp(V(x)+\sum_{k=1,3,5} t_kx^k).
\end{align}
Here $V(x)$ is required to be indepedent of $t$ and is chosen to ensure the convergency of the integral. Then the elements $a_{i,j}$ in the matrix $A$ can be rewritten as
\begin{align*}
a_{i,j}=\int_{-\infty}^{t_1}\phi_i\phi_jdt_1=\int_{-\infty}^{t_1}\left(\iint_{\mathbb{R}_+^2}x^iy^j\alpha(x;t)\alpha(y;t)dxdy\right)dt=\iint_{\mathbb{R}_+^2}\frac{x^iy^j}{x+y}\alpha(x;t)\alpha(y;t)dxdy,
\end{align*}
which is nothing but the bimoments of sCBOPs. Therefore, by using the Heine's formula (\ref{heine}), the $\tau$-function $f$ can be transformed into matrix integrals. In this sense, we have obtained matrix integral solutions to the CKP equation.

From Sato theory, we know that $\tau$-functions of a hierarchy of soliton equations share similar properties of its members. Next, we would like to show how to construct matrix integral solutions to the CKP hierarchy starting from the CKP equation.
\begin{prop}\label{JM}
If we introduce the neutral bose fields $\xi_j(j\in\mathbb{Z})$ such that
\begin{align*}
[\xi_i,\xi_j]=(-1)^j\delta_{i+j,-1},
\end{align*}
where $\xi_i$ ($i<0$) are the annihilation operators and $\xi_i$ ($i\geq0$) are the creation operators, then the Hamiltonian
\begin{align*}
H(t_+)=\frac{1}{2}\sum_{i=1,odd}^{+\infty}\sum_{j\in\mathbb{Z}}(-1)^{j-1}t_i\xi_j\xi_{-i-j-1}
\end{align*}
will give the $\tau$-function of CKP hierarchy by
\begin{align*}
\tau(t)=\langle e^{H(t_+)}g\rangle^{-2}
\end{align*}
with arbitrary group-like elements $g$.
\end{prop}
This proposition is firstly given in \cite{DJKM}, and a renewed version can be referred to \cite{VOS}. From this proposition, we know that if $\tau(t)$ is a $\tau$-function of the CKP equation, then it is also the $\tau$-function of the CKP hierarchy after introducing higher-order time flows which are compatible with the CKP equation. A suitable choice of higher order time flow is set by
\begin{align*}
\alpha(x;t)=\sum_{k=1,odd}^{+\infty}t_kx^k.
\end{align*}
Therefore,
we have the following proposition.
\begin{prop}
The $\tau$-function of the CKP hierarchy has the determinant solution
\begin{align*}
\tau_N=\det\left(\iint_{\mathbb{R}_+^2}\frac{x^iy^j}{x+y}\alpha(x;t)\alpha(y;t)dxdy\right)_{0\leq i,j\leq N-1}.
\end{align*}
Or equivalently, it has the matrix integral form
\begin{align*}
\tau_N=\int_{\mathbb{T}}\left(\Delta_{N}(X)\Delta_{N}(Y)\right)^2\prod_{i,j=1}^N\frac{1}{x_i+y_j}\alpha(x_i;t)\alpha(y_j;t)dx_idy_j,
\end{align*}
where $\mathbb{T}$ is a subset of $\mathbb{R}_+^N\times\mathbb{R}_+^N$ indicated before.
\end{prop}

\begin{proof}
From Proposition \ref{PPCKP}, we know that
\begin{align*}
\tau_N=\det\left(\iint_{\mathbb{R}_+^2}\frac{x^iy^j}{x+y}\alpha(x;t)\alpha(y;t)dxdy\right)_{0\leq i,j\leq N-1}.
\end{align*}
is the solution of the CKP equation. According to Proposition \ref{JM}, $\tau_N$ also gives the $\tau$-function of the CKP hierarchy. Moreover, the seed functions $\phi_i$ in (\ref{sm}) satisfy the dispersion relations
\begin{align*}
\frac{\partial \phi_i}{\partial t_m}=\frac{\partial^m \phi_i}{\partial t_1^m}, \quad m\in2\mathbb{Z}+1.
\end{align*}
Thus the $\tau$-function of the CKP hierarchy is indeed the same as the one of the CKP equation with higher-order time flows.
\end{proof}

\subsection{The Cauchy two-matrix model and integrable systems}
In \cite{BGS3}, the authors proposed the concept of the Cauchy two-matrix model related to a special determinant point process with the Cauchy kernel. The correlation function of this model is defined by
\begin{align*}
&\mathcal{R}^{r,k}(x_1,\cdots,x_r;y_1,\cdots,y_k)=\\
&\frac{N!\prod_{j=1}^r\alpha(x_j)\prod_{j=1}^k\beta(y_j)}{(N-r)!(N-k!)\mathcal{Z}_N}\int \prod_{l=r+1}^N\alpha(x_l)dx_l\prod_{j=k+1}^N\beta(y_j)dy_j\Delta_N(X)\Delta_N(Y)\det[K(x_i,y_j)]_{1\leq i,j\leq N},
\end{align*}
where $K(x,y)$ is the Cauchy kernel in the form of $K(x,y)=\frac{1}{x+y}$ and $\mathcal{Z}_N$ is the normalization constant. Here we would like to mention that
the $(r,k)$-point correlation function allows one to compute the probability of having $r$ eigenvalues of the first matrix and $k$ eigenvalues of the second matrix in measurable sets of the real axis. Usually, we call the normalization constant $\mathcal{Z}_N$ the partition function which has the form
\begin{align*}
\mathcal{Z}_N=\frac{1}{N!}\int_{\mathbb{R}^N\times \mathbb{R}^N}\Delta_N(X)\Delta_N(Y)\det[K(x_i,y_j)]_{1\leq i,j\leq N}\alpha(X)\beta(Y)dXdY.
\end{align*}
Indeed, by using the Cauchy determinant formula
\begin{align*}
\det[K(x_i,y_j)]_{1\leq i,j\leq N}=\frac{\Delta_N(X)\Delta_N(Y)}{\prod_{i,j=1}^N(x_i+y_j)},
\end{align*}
we find that $\mathcal{Z}_N$ is exactly the matrix integral representation of the $\tau$-function of the CKP hierarchy if we impose the symmetric reduction $\alpha(X)=\beta(X)$ on $\mathcal{Z}_N$ and the time evolutions satisfy $\alpha(x;t)=e^{xt}\alpha(x;0)$.

According to the underlying infinite dimensional Lie algebras, integrable systems are classified into the KP hierarchy, the BKP hierarchy, the CKP hierarchy and the DKP hierarchy.  So far, the corresponding matrix models have all been found. The following table demonstrates this fact.
\begin{table}[htbp]\small
\caption{Matrix models related to integrable hierarchies}
\centering
\begin{tabular}{cc}\hline\hline
{Integrable Hierarchies} & {Ensembles}\\\hline
KP hierarchy & Hermite ensemble (2d-gravity matrix models)\\
BKP hierarchy & Bures ensemble (2d-quantum gravity matrix models)\\
CKP hierarchy & Cauchy ensemble (Cauchy two-matrix models)\\
DKP hierarchy & Gaussian Orthogonal/Symplectic ensemle\\
\hline\hline
\end{tabular}
\label{dyandalg}
\end{table}

\section{From Cauchy to Bures---A version of integrable systems}
In \cite{BGS3}, the authors realized there may exist some connections between the Cauchy two-matrix model and Bures ensemble on the level of correlation functions. Forrester and Kieburg gave the first explaination upon this topic in \cite{FK}. They noticed that the correlation function of Cauchy two-matrix model can be expressed in terms of determinants and that of Bures ensemble can be expressed in Pfaffians. Based on the relations between determinants and Pfaffians, they found the connection between Cauchy two-matrix model and Bures ensemble and thus giving the connection between determinant point process and Pfaffian point process. In this section, we will give another explaination of the relationship between the Cauchy two-matrix model and Bures ensemble from the view of point of integrable systems based on the reductive theory proposed by Jimbo and Miwa \cite{JM}.

 As one of the members of the BKP hierarchy, the following equation
\begin{equation}\label{BKP}
(D_{t_1}^6-5D_{t_1}^3D_{t_3}-5D_{t_3}^2+9D_{t_1}D_{t_5})\tau_N\cdot\tau_N=0
\end{equation}
shares the same matrix integral solutions as the BKP hierarchy. As is known, matrix integral solutions to the BKP hiearchy are given by \cite{HL}
\begin{align}
\tau_N=\int_{\mathbb{R}_+^{N}}\prod_{1\leq i<j\leq N}\frac{(x_i-x_j)^2}{x_i+x_j}\prod_{i=1}^{N}\alpha(x_i;t)dx_i
\end{align}
with $\alpha(x;t)$ has the same form in (\ref{sm}).
Then by applying the reduction theory in integrable hierarchies,
the Sawada-Kotera equation
\begin{align}\label{SK}
(D_{t_1}^6+9D_{t_1}D_{t_5})\tau_N\cdot\tau_N=0
\end{align}
is obtained as the 3-reduction of the BKP equation. Note that this equation is obtained from BKP equation (\ref{BKP}) by eliminating the terms with $D_3$ directly which means to impose the following additional constraints in terms of $\tau$-function as
\begin{align*}
\frac{\partial}{\partial t_{3j}}\tau_N(t)=0,\quad j=1,2,\cdots.
\end{align*}
This reduction confines the Lie algebra of the transformation group of Sawada-Kotera to a subalgebra of $go(\infty)$ as an affine Lie algebra $A_2^{(2)}$. Therefore, $$\tau_N=\int_{\mathbb{R}_+^{N}}\prod_{1\leq i< j\leq N}\frac{(x_i-x_j)^2}{x_i+x_j}\prod_{i=1}^{N}\exp(V(x_i)+x_it_1+x_i^5t_5)dx_i$$
is indeed the $\tau$-function of the Sawada-Kotera equation.

As is discussed in previous section, under certain assumptions, the partition function of the Cauchy two-matrix model gives the $\tau$-function of the CKP hierarchy. The $\tau$-function to the CKP hierarchy reads as
\begin{align}
\tau_N'=\int_{\mathbb{R}_+^N\times\mathbb{R}_+^N}\frac{\prod_{1\leq i,j\leq N}(x_i-x_j)^2\prod_{1\leq i,j\leq N}(y_i-y_j)^2}{\prod_{i,j=1}^N(x_i+y_j)}\prod_{i=1}^N\alpha(x_i;t)\alpha(y_i;t)dx_idy_i
\end{align}
with $\alpha(x;t)$ given in (\ref{sm}).
Also by applying the reduction theory in integrable systems, the CKP equation (\ref{CKP}) can be reduced to Kaup-Kuperschmidt equation
\begin{align}\label{KK1}
\left\{
\begin{aligned}
&D_{t_1}^4\tau_N'\cdot \tau_N'+6\tau_N'\sigma=0,\\
&D_{t_1}(D_{t_1}^5+144D_{t_5})\tau_N'\cdot\tau_N'-90D_{t_1}^2\tau_N'\cdot\sigma=0.
\end{aligned}
\right.
\end{align}
by imposing the additional constraint on the $\tau$-function as
\begin{align*}
\frac{\partial}{\partial t_{3j}}\tau'_N(t)=0,\quad j=1,2,\cdots
\end{align*}
Here we would like to remark that $\sigma$ is an auxiliary function in \eqref{KK1} which can be eliminated. Therefore it is not essential to consider $\sigma$. And the Lie algebra of the transformation group for the Kaup-Kuperschmidt hierarchy is a subalgebra of $sp(\infty)$ which is isomorphism to the affine Lie algebra $A_2^{(2)}$ and
\begin{align*}
\tau_N'=\int_{\mathbb{R}_+^N\times\mathbb{R}_+^N}\frac{\prod_{1\leq i,j\leq N}(x_i-x_j)^2\prod_{1\leq i,j\leq N}(y_i-y_j)^2}{\prod_{i,j=1}^N(x_i+y_j)}\prod_{i=1}^N\tilde{\alpha}(x_i;t)\tilde{\alpha}(y_i;t)dx_idy_i
\end{align*}
with $\tilde{\alpha}(x;t)=\exp(V(x)+xt_1+x^5t_5)$ being the $\tau$-function of Kaup-Kuperschmidt equation.

It is pointed out that there exists an integrable system called the modified Kaup-Kuperschmidt equation which links two slightly different Sawada-Kotera equation and Kaup-Kuperschmidt equation \cite{HRN}. For the Sawada-Kotera equation \eqref{SK} and the Kaup-Kuperschmidt equation \eqref{KK1} considered here, they are linked by the following modified Kaup-Kuperschmidt equation
\begin{align}\label{mKK1}
9v_{t_5}-5(v_{t_1}v_{3t_1}+v_{t_1t_1}^2+v_{t_1}^3+4vv_{t_1}v_{t_1t_1}+v^2v_{3t_1}-v^4v_{t_1})+v_{5t_1}=0.
\end{align}
By taking the dependent variable transformation
$
v=-3\left(\log({\tau'_N}/{\tau_N^2})\right)_{t_1}
$ with $\tau_N$ and $\tau'_N$ being the $\tau$-functions of Sawada-Kotera equation and Kaup-Kuperschmidt equation respectively,
equation \eqref{mKK1} can be transformed into
\begin{align}\label{mKK2}
\left\{
\begin{aligned}
&\left(9D_{2,t_5}+\frac{1}{6}D_{2,t_1}^5\right)\tau_N'\cdot\tau_N=0,\\
&D_{2,t_1}^2\tau_N'\cdot\tau_N=0,
\end{aligned}
\right.
\end{align}
where $D_{m,x}^jf(x)\cdot g(x)=\frac{\partial^j}{\partial s^j}f(x+s)g(x-ms)|_{s=0}$.
This reveals the fact that  that time-dependent partition functions of Bures ensemble and the Cauchy two-matrix model can be linked by the integrable system \eqref{mKK2}.

\section{Conclusion and Discussion}
In this paper, we mainly focus on the integrable hierarchy related to Cauchy two-matrix mode. A method to find out this connection is based on the relationship between orthogonal polynomials theory and integrable systems. To start with, we firstly consider a symmetric reduction of CBOPs. By introducing time deformation in CBOPs, a Toda-type equation along with its Lax pair are obtained. Then we find out that the time-dependent partition function of Cauchy two-matrix model (or Cauchy ensemble) is related to the CKP hierarchy according to Sato's $\tau$-function theory. Therefore, we give a picture of the relations between a certain integrable hierarchy and corresponding statistic model. By noting the fact that there exists an integrable system connecting Sawada-Kotera equation with Kaup-Kuperschmidt equation whose $\tau$-functions are the partition functions of Bures ensemble and Cauchy ensemble respectively, we give an explanation of the relation between Bures ensemble and Cauchy ensemble in the perspective of integrable system.

However, there is still a lot of work remaining to study. For instance, we are now able to apply the Borel-Virasoro algebra to the partition function of Cauchy ensemble to obtain the Virasoro constraints of CKP hierarchy. Although there has been some work upon the Virasoro constraintes of CKP hierarchy in literature \cite{WL}, it is worth studying matrix models and corresponding integrable hierarchies from the viewpoint of vertex operators. Furthermore, although we have found out all the main 2+1 dimensional integrable hierarchies and their corresponding matrix models, we wonder whether there exist some other types of matrix models connected with 1+1 dimensional integrable hierarchies such as the relation between Kontsevich integral and the KdV equation and what the algebraic and geometric structures are underlying these matrix models?

\section*{Acknowledgement}
This work was supported by National Natural Science Foundation of China (Grant No. 11271266, No.11705284 and No. 11701550) and Beijing Natural Science Foundation (Grant No. 1162003). Dr. C. X. Li would like to thank for the hospitality of School of Mathematics and Science during her visit to Fudan University. S. H. Li would like to thank Dr. X. K. Chang and Mr. B. Wang for helpful discussions and thank Prof. X. B. Hu for his attentive guidance.

\begin{appendix}
\section{Proofs of Propositions 3.1 and Proposition 3.3 by Pfaffians}
In this appendix, we would like to give a brief review on some known facts of Pfaffians and then prove Proposition 3.1 and Proposition 3.3 by pfaffian techniques. It turns out that the C-Toda lattice is nothing but pfaffian identities with the $\tau$-functions given by determinants.

The term Pfaffian was introduced by Arthur Cayley in 1852, who named it after Johann Friedrich Pfaff. Pfaffian is generally used in theoretical physics nowadays. Let us first have a look at the definition of a pfaffian. Let $A=(a_{i,j})_{1\leq i,j\leq 2N}$ be a $2N\times 2N$ skew-symmetric matrix. The pfaffian of $A$, that is, $Pf(A)$ is defined as
\begin{align*}
Pf(A)&=Pf(a_{i,j})_{1\leq i,j\leq 2N}=Pf\left[
\begin{array}{cccc}
0&a_{1,2}&\cdots&a_{1,{2N}}\\
-a_{1,2}&0&\cdots&a_{2,{2N}}\\
\vdots&\vdots&\ddots&\vdots\\
-a_{1,{2N}}&-a_{2,{2N}}&\cdots&0
\end{array}
\right]\\
&=\sum {'sgn}\left(
\begin{array}{cccc}
1&2&\cdots&2N\\
j_1&j_2&\cdots&j_{2N}
\end{array}
\right)a_{j_1j_2}a_{j_3j_4}\cdots a_{j_{2N-1}j_{2N}},
\end{align*}
where $\sum'$ means the sum over all possible combinations of pairs selected from $\{1,2,\cdots,2N\}$ satisfying $j_1<j_2$, $\cdots$, $j_{2N-1}<j_{2N}$ and $j_1<j_3<\cdots<j_{2N-1}$. As an equivalent definition, an $n$-th order pfaffian $pf(1,2,\ldots, 2N)$ can be expanded as
\begin{eqnarray*}
pf(1,2,\ldots, 2N)&=&\sum_{j=2}^{2N}(-1)^jpf(1,j)pf(2,3,\ldots,\hat{j},\ldots,2N)
\end{eqnarray*}
where $\hat{j}$ means that the index $j$ is omitted. By this formula, the pfaffian $pf(1,2,\ldots, 2N)$ may be recursively defined when pfaffian entries $pf(i,j)$ are given. It is obvious that $pf(1,2,\ldots, 2N)=Pf(A)$ when $(i,j)=a_{i,j}$. In this sense, the above-mentioned two definitions can be unified.

In the remaining part, we would like to adopt the notation below to denote a pfaffian
\begin{eqnarray*}
\centering
pf(i_1,i_2,\cdots,i_{2N})=Pf\left[
\begin{array}{cccc}
0&a_{i_1,i_2}&\cdots&a_{i_1,i_{2N}}\\
-a_{i_1,i_2}&0&\cdots&a_{i_2,i_{2N}}\\
\vdots&\vdots&\ddots&\vdots\\
-a_{i_1,i_{2N}}&-a_{i_2,i_{2N}}&\cdots&0
\end{array}
\right],
\end{eqnarray*}
where the pfaffian elements $pf(i_j,i_k)=a_{i_j,i_k}$.

It is noted that any $n$-th order determinant can be expressed as an $n$-th order pfaffian. Therefore, if we define pfaffian entries by
\begin{equation*}
pf(i,j)=pf(i^*,j^*)=0,\ \ pf(i,j^*)=I_{i,j},
\end{equation*}
then $\tau_n$ and $\tilde{\tau}_n$ given before can be also expressed by means of Pfaffians as
 \begin{align*}
&\tau_n=pf(0,1,\ldots,n-1,n-1^*,\ldots,1^*,0^*),\\
&\tilde{\tau}_n=pf(0,1,\ldots,n-2,n,n-1^*,\ldots,1^*,0^*).
\end{align*}

\begin{prop}
For the pfaffian $\tau_n=pf(0,1,\cdots,n-1,n-1^*,\cdots,1^*,0^*)$ defined above, if its pfaffian entries satisfy the relation
\begin{align*}
\frac{d}{dt}pf(i,j^*)=pf(i,j+1^*)+pf(i+1,j^*),
\end{align*}
then we have
\begin{align}
\frac{d}{dt}\tau_n&=\frac{d}{dt}pf(0,\cdots,n-1,n-1^*,\cdots,0^*)\nonumber\\
&=2pf(0,\cdots,n-2,n,n-1^*,\cdots,0^*)\nonumber\\
&=2\tilde{\tau}_n.\label{A1}
\end{align}
\end{prop}
\begin{proof}
Let us first prove the following equality
\begin{align}\label{ind}
\frac{d}{dt}pf(i_1,\cdots,i_N,j_1^*,\cdots,j_N^*)&=\sum_{k=1}^Npf(i_1,\cdots,i_k+1,\cdots,i_N,j_1^*,\cdots,j_N^*)\nonumber\\
&+\sum_{k=1}^Npf(i_1,\cdots,i_N,j_1^*,\cdots,j_k+1^*,\cdots,j_N^*)
\end{align}
with the pfaffian entries defined by
\begin{eqnarray*}
&&pf(i_k,i_l)=pf(j_k^*,j_l^*)=0, \\
&&\frac{d}{dt}pf(i_k,j_l^*)=pf(i_k,j_l+1^*)+pf(i_k+1,j_j^*).
\end{eqnarray*}

In what follows, we are going to prove \eqref{ind} by induction. It is obvious that \eqref{ind} is true for $N=1$.
Assume that \eqref{ind} holds for $N$. For $N+1$, we have
\begin{align*}
\frac{d}{dt}pf(i&_1,\cdots,i_N,i_{N+1},j_1^*,\cdots,j_N^*,j_{N+1}^*)\\
&=\frac{d}{dt}[\sum_{l=1}^{N+1}(-1)^{N+l}pf(i_{N+1},j_l^*)pf(i_1,\cdots,i_N,j_1^*,\cdots,\hat{j}_l^*,\cdots,j_{N+1}^*)]\\
&=\sum_{l=1}^{N+1}(-1)^{N+l}\left[pf(i_{N+1}+1,j_l^*)+pf(i_{N+1},j_{l}+1^*)\right]pf(i_1,\cdots,i_N,j_1^*,\cdots,\hat{j}_l^*,\cdots,j_{N+1}^*)\\
&+\sum_{l=1}^{N+1}(-1)^{N+l}pf(i_{N+1},j_l^*)[\sum_{k=1}^Npf(i_1,\cdots,i_k+1,\cdots,i_N,j_1^*,\cdots,\hat{j}_l^*,\cdots,j_{N+1}^*)\\
&+\sum_{k=1,k\not=l}^{N+1}pf(i_1,\cdots,i_N,j_1,\cdots,j_k+1^*,\cdots,\hat{j}_l^*,\cdots,j_{N+1}^*)]\\
&=\sum_{k=1}^{N+1}pf(i_1,\cdots,i_k+1,\cdots,i_{N+1},j_1^*,\cdots,j_{N+1}^*)+\sum_{k=1}^{N+1}pf(i_1,\cdots,i_{N+1},j_1^*,\cdots,j_k+1^*,\cdots,j_{N+1}^*),
\end{align*}

So far, we have completed the proof of \eqref{ind}. Notice that $pf(0,\cdots,n-2,n,n-1^*,\cdots,0^*)=pf(0,\cdots,n-1,n^*,n-2^*,\cdots,0^*)$ due to the symmetry $I_{i,j}=I_{j,i}$. Simply by taking $\{i_1,\cdots,i_N\}$ as $\{0,\cdots,n-1\}$ and $\{j_1^*,\cdots,j_N^*\}$ as $\{n-1^*,\cdots,0^*\}$ in \eqref{ind}, we have \eqref{A1}.
\end{proof}

Although Pfaffians may be obtained from antisymmetric determinants, their properties are more varied than those of determinants. Determinantal identities such as Pl\"ucker relations and Jacobi identities, are extended and unified as pfaffian identities which are very useful in integrable systems. Here we list two most useful identities:
\begin{align*}
&pf(a_1,a_2,a_3,a_4,1,\ldots,2n)pf(1,2,\ldots,2n)
=\sum_{j=2}^{4}(-1)^jpf(a_1,a_j,1,\ldots,2n)pf(a_2,\hat{a}_j,a_{4},1,\ldots,2n),\\&pf(a_1,a_2,a_3,1,\ldots,2n-1)pf(1,2,\ldots,2n)
=\sum_{j=1}^{3}(-1)^{j-1}pf(a_j,1,\ldots,2n-1)pf(a_1,\hat{a}_j,a_{3},1,\ldots,2n).
\end{align*}
As an application, we will demonstrate how to use pfaffian techniques to prove that $\tau_n$ and $\sigma_n$ are solutions to the C-Toda lattice (\ref{CT}).

\begin{prop}
The C-Toda lattice \eqref{CT} has solutions
\begin{align*}
&\tau_n=\begin{vmatrix}I_{0,0}&\cdots&I_{0,n-1}\\
\vdots&&\vdots\\
I_{n-1,0}&\cdots&I_{n-1,n-1}\end{vmatrix}=pf(0,\cdots,n-1,n-1^*,\cdots,0^*),\\
&\sigma_n=\begin{vmatrix}I_{0,0}&\cdots&I_{0,n}\\ \vdots&&\vdots\\I_{n-1,0}&\cdots&I_{n-1,n}\\w_0&\cdots& w_n\end{vmatrix}=(-1)^npf(d_0,0,\cdots,n-1,n^*,\cdots,0^*)
\end{align*}
with pfaffian entries defined by
\begin{align*}
&pf(d_0,i)=pf(d_0^*,i^*)=pf(d_0,d_0^*)=0,\, pf(d_0,i^*)=pf(d_0^*,i)=w_i, \\
&\frac{d}{dt}pf(i,j^*)=pf(i+1,j^*)+pf(i,j+1^*)=pf(d_0,d_0^*,i,j^*),\, \frac{d}{dt}pf(d_0,i^*)=pf(d_0,i+1^*)
\end{align*}
which correspond to the conditions $\frac{d}{dt}I_{i,j}=I_{i+1,j}+I_{i.j+1}=\omega_i\omega_j$, $\frac{d}{dt}\omega_i=\omega_{i+1}$.
\end{prop}
\begin{proof}
By using derivative formulae for Pfaffians repeatedly, it is easy to derive that
\begin{align*}
\tau_{n,t}&=pf(d_0,d_0^*,0,\cdots,n-1,n-1^*,\cdots,0^*)\\
&=2pf(0,\cdots,n-2,n,n-1^*,\cdots,0^*),\\
\tau_{n,tt}&=2pf(d_0,d_0^*,0,\cdots,n-2,n,n-1^*,\cdots,0^*).
\end{align*}
Substituting the above results into the C-Toda lattice \eqref{CT}, we obtain the following two expressions
\begin{align*}
&pf(d_0,d_0^*,0,\cdots,n,n^*,\cdots,0^*)pf(0,\cdots,n-1,n-1^*,\cdots,0^*)\\-&pf(0,\cdots,n,n^*,\cdots,0^*)pf(d_0,d_0^*,0,\cdots,n-1,n-1^*,\cdots,0^*)\\
=&pf(d_0,0,\cdots,n-1,n^*,\cdots,0^*)pf(d_0^*,0,n,n-1^*,\cdots,0^*),\\
&pf(d_0,d_0^*,0,\cdots,n-1,n^*,n-2^*,\cdots,0^*)pf(0,\cdots,n-1,n-1^*,\cdots,0^*)\\
-&pf(d_0,d_0^*,0,\cdots,n-1,n-1^*,\cdots,0^*)pf(0,\cdots,n-1,n^*,n-2^*,\cdots,0^*)\\
=&pf(d_0,0,\cdots,n-1,n^*,\cdots,0^*)pf(d_0^*,0,\cdots,n-1,n-2^*,\cdots,0^*),
\end{align*}
which are indeed two special cases of pfaffian identities mentioned above.
\end{proof}

\end{appendix}

\small
\bibliographystyle{abbrv}
\def\cydot{\leavevmode\raise.4ex\hbox{.}}
  \def\cydot{\leavevmode\raise.4ex\hbox{.}} \def\cprime{$'$}

\end{document}